\documentclass[runningheads,envcountsame, envcountsect]{llncs}
\usepackage[utf8]{inputenc}

\usepackage{amsmath}

\usepackage{amsthm}
\usepackage{amssymb}
\usepackage{bbm}
\usepackage{breakcites}
\usepackage[english]{babel}
\usepackage{color}
\usepackage{MnSymbol}
\usepackage{enumitem}
\usepackage{ellipsis}
\usepackage[bookmarks=false]{hyperref}
\usepackage[misc,geometry]{ifsym} \usepackage{dsfont}
\usepackage{mathrsfs}
\usepackage{mathtools}
\usepackage{algorithmic}
\usepackage[linesnumbered,ruled,vlined]{algorithm2e}
\usepackage{etoolbox} \usepackage{stmaryrd}
\usepackage{thmtools}
\usepackage{thm-restate}
\usepackage{subcaption}
\usepackage{systeme}
\usepackage{multirow}
\usepackage{tabularx}

\usepackage{hyperref}
\usepackage{xurl}
\usepackage{cleveref}

\usepackage{tikz}
\usepackage{accents}
\usetikzlibrary{patterns}
\usetikzlibrary{matrix}

\usepackage{todonotes}
\usepackage{lipsum}

 \newtoggle{llncs}

\makeatletter
\@ifclassloaded{llncs}{\toggletrue{llncs}}{\togglefalse{llncs}}
\makeatother

\iftoggle{llncs}{}{
	\usepackage{a4wide}
	\usepackage{amsthm}
	\usepackage[foot]{amsaddr}		
}

\newcommand{\mat}[1]{\ensuremath{\boldsymbol{\rm #1}}}

\newcommand{\prob}[2]{\mathbb{P}_{#1}\left(#2\right)}

\newcommand{\expect}[2]{\mathbb{E}_{#1}\left(#2\right)}

\newcommand{\IInt}[2]{\left\llbracket #1, #2 \right\rrbracket}

\newcommand{\fract}[2]{\hbox{\leavevmode
\kern.1em \raise .5ex \hbox{\the\scriptfont0 $#1$}\kern-.1em }/
\hbox{\kern-.15em \lower .25ex \hbox{\the\scriptfont0 $#2$}}
}

\newcommand{\Psucc}{\ensuremath{p_{\mathrm{succ}}}}

\newcommand{\Fq}{\mathbb{F}_q}

\newcommand{\F}{\mathbb{F}}
\newcommand{\NN}{\mathbb{N}}

\newcommand{\C}{\mathcal{C}}

\newcommand{\dgv}{d_{\mathsf{GV}}}

\newcommand{\hw}[1]{\ensuremath{\left|#1\right|}}
\newcommand{\support}[1]{\ensuremath{\mathrm{supp}\left(#1\right)}}
\newcommand{\card}[1]{\ensuremath{\left| #1 \right|}}

\newcommand{\transp}[1]{\ensuremath{#1^\intercal}}

\newcommand{\class}[1]{\ensuremath{\left[#1\right]}}
\newcommand{\rep}[1]{\ensuremath{\overline{#1}}}
\newcommand{\repbis}[1]{\ensuremath{\widehat{#1}}}
\newcommand{\done}{\ensuremath{\mathsf{Done}}}

\newcommand{\bigslant}[2]{{\raisebox{.15em}{$#1$}\mathord{\left/\raisebox{-.15em}{$#2$}\right.}}}
\newcommand{\projspace}[1]{\ensuremath{\bigslant{#1}{\sim}}}

\newcommand{\sA}{{\mathscr{A}}}

\newcommand{\sL}{{\mathscr{L}}}

\newcommand{\oo}[1]{\ensuremath{\mathop{}\mathopen{}o\mathopen{}\left(#1\right)}}

\newcommand{\Gm}{{\mathbf{G}}}
\newcommand{\Hm}{{\mathbf{H}}}
\newcommand{\Am}{{\mathbf{A}}}
\newcommand{\Mm}{{\mathbf{M}}}
\newcommand{\Pm}{{\mathbf{P}}}
\newcommand{\Rm}{{\mathbf{R}}}

\newcommand{\Id}{{\mathbf{I}}}
\newcommand{\Ov}{{\mathbf{0}}}

\newcommand{\cv}{{\mathbf{c}}}

\newcommand{\ev}{{\mathbf{e}}}

\newcommand{\sv}{{\mathbf{s}}}

\newcommand{\uv}{{\mathbf{u}}}
\newcommand{\vv}{{\mathbf{v}}}

\newcommand{\xv}{{\mathbf{x}}}
\newcommand{\yv}{{\mathbf{y}}}
\newcommand{\zv}{{\mathbf{z}}}

\newcommand{\eqdef}{\stackrel{\textrm{def}}{=}}
\renewcommand{\leq}{\leqslant}

\renewcommand{\geq}{\geqslant}

\newcommand{\SD}{\textsf{SD}\xspace}
\newcommand{\SDitH}{\textsf{SDitH}\xspace}
\newcommand{\Wave}{\textsf{Wave}\xspace}
\newcommand{\BIKE}{\textsf{BIKE}\xspace}
\newcommand{\McEliece}{\textsf{Classic McEliece}\xspace}

\iftoggle{llncs}{

}{

}

\iftoggle{llncs}{

\newtheorem{fact}[theorem]{Fact}
\newtheorem{notation}[theorem]{Notation}
\newtheorem{constraint}[theorem]{Constraint}
\newtheorem{model}[theorem]{Model}
\newtheorem{assumption}[theorem]{Assumption}

\counterwithin{equation}{section}
\counterwithin{figure}{section}
\counterwithin{table}{section}

\spnewtheorem{pb}[theorem]{Problem}{\bfseries}{\itshape}

\usepackage{thm-restate}
}{

\newtheorem{remark}{Remark}
\newtheorem{lemma}{Lemma}

\newtheorem{definition}{Definition}
\newtheorem{theorem}{Theorem}

\newtheorem{pb}{Problem}
\newtheorem{notation}{Notation}

}

\newcommand{\ie}{{\em i.e. }}

\begin{document}
	
\title{Projective Space Stern Decoding  and Application to SDitH}

\iftoggle{llncs}{

\author{Kevin Carrier\inst{1} \and Valerian Hatey\inst{1} \and Jean-Pierre Tillich\inst{2}}
\authorrunning{K. Carrier \and V. Hatey \and J-P. Tillich}

\institute{
	ETIS UMR 8051 - Cergy-Paris Université, ENSEA, CNRS, \email{kevin.carrier@cyu.fr,valerian.hatey@ensea.fr}
	\and
	Project COSMIQ, Inria de Paris, \email{jean-pierre.tillich@inria.fr}
	\thanks{The work of KC, VH and JPT was funded by the French Agence Nationale de la Recherche through ANR JCJC DECODE (ANR-22-CE39-0004-01) for KC and VH and ANR-22-PETQ-0008 PQ-TLS for JPT.} 	
}

}{

\author{K\'{e}vin Carrier$^{1}$} 
\email{kevin.carrier@cyu.fr}
\author{Val\'{e}rian Hatey$^{1}$} 
\email{valerian.hatey@ensea.fr}  
\author{Jean-Pierre Tillich$^{4}$} 
\email{jean-pierre.tillich@inria.fr}

\address{$^{1}$ ETIS UMR 8051 - Cergy-Paris Université, ENSEA, CNRS}
\address{$^{2}$ Project COSMIQ, Inria de Paris}

\thanks{The work of KC and JPT was funded by the French Agence Nationale de la Recherche through project ANR JCJC DECODE (ANR-22-CE39-0004-01) for KC and by 
project ANR-22-PETQ-0008 PQ-TLS for JPT.} 	

}
\maketitle

\begin{abstract}
We show that here standard decoding algorithms for generic linear codes over a finite field can  speeded up by a factor which is essentially the size of the finite field
by reducing it to a low weight codeword problem and working in the relevant projective space. We apply this technique to \SDitH and show that the parameters of both the original 
submission and the updated version fall short of meeting the security requirements asked by the NIST.
\end{abstract}

\section{Introduction} 
\label{sec:intro}

Code-based cryptography is based on the hardness of the decoding problem. In its syndrome (and fixed weight) version it is given for the Hamming metric
(where we denote the Hamming weight of a vector $\xv$ by $\hw{\xv}$) by
\begin{pb}[Syndrome Decoding $\SD(\Hm,\sv,t)$]
\label{pb:SD}
Given a matrix $\Hm \in \Fq^{(n-k) \times n}$, a syndrome $\sv \in \Fq^{n-k}$ and a weight $t \in \IInt{0}{n}$, find a vector $\ev \in \Fq^{n}$ such that $\Hm \ev = \sv$ and $\hw{\ev} = t$.
\end{pb}
In other words, it consists in solving a linear system with a constraint on the weight of the solution. This non-linear constraint is commonly believed to make the problem difficult on average over $\Hm$ for suitable values of $t$.
Despite that many efforts have been spent over the last 60 years \cite{P62,S88,D91,MMT11,BJMM12,MO15,BM17,BM18,CDMT22}, the problem remains hard in the range of parameters given above, even with the help of a quantum computer 
\cite{B10,KT17}.
Thus, the decoding problem has raised interest among cryptosystem designers. It is today the heart of the security of PKE and signature schemes submitted to the NIST competitions\footnote{\url{https://csrc.nist.gov/projects/post-quantum-cryptography}} such as \McEliece \cite{AABBBBDGGGGMPRSTVZ22}, \BIKE \cite{ABCCGLMMMNPPPSSSTW20}, \Wave \cite{BCCCDGKLNSST23} and \SDitH \cite{AFGGHJJPRRY23}.
It is quite common to study the binary version of the decoding problem but the non-binary case also aroused interest \cite{BLP10,BLP11a} or more recently with the signature schemes \Wave \cite{DST19a} or \SDitH \cite{FJR22a} for instance. The security of \Wave is based on ternary codes and \SDitH addresses the syndrome decoding over the fields $\F_{256}$ and $\F_{251}$. In this article, we focus on the case where $q$ is large, as is the case in\SDitH.

The best known decoding algorithm are Information Set Decoding (ISD) initiated by Prange in \cite{P62}. The idea of Prange basically consists in 
guessing that $\ev$ is zero on an \emph{information set}, that is a set of $k$ positions that determines the whole vector $\ev$ considering the linear relation $\Hm \ev = \sv$. If the guess does not allow to find a vector of weight $t$, then we repeat the process changing the information set until we make the right guess on it.

There have been numerous improvements to the Prange algorithm. One of the first breakthrough in this domain uses the \emph{birthday paradox} \cite{S88,D91}. Basically, the ISD template is used here to reduce the decoding problem to a collision search. Later, other techniques were introduced to improve ISD. For instance, \cite{MMT11} and \cite{BJMM12} exploit the fact that a low weight vector can be represented in multiple ways as the sum of two lower weight vectors, 
it is the so-called \emph{representation technique} introduced by Howgrave-Graham and Joux in \cite{HJ10}. In the \SDitH specifications, it is noticed that the representation technique was originally designed for the binary case and that it loses its interest when $q$ is large. This claim is supported by Meurer in his PhD thesis \cite{M12} and also by Canto-Torres in \cite{C17a} where he shows that the MMT \cite{MMT11} and BJMM \cite{BJMM12} complexity exponent tend to the Prange complexity exponent when $q$ tends to infinity. This is the reason why the algorithm on which \SDitH focuses on is Stern \cite{S88} which is considered optimal by the authors of \SDitH in their particular context.

\subsection{Our contribution}
Our main observation here is that decoding over a big field can basically be speeded up by a factor which is the size of the field by a simple homogenizing trick (or what is the same by a reduction to the low weight codeword search problem). The idea is that instead of looking for a vector $\ev$ of weight $t$ satisfying $\Hm \ev = \sv$ we look for a vector $\xv$ of weight 
$t$ such that $\Hm \xv$ is {\em proportional} to $\sv$, {\em i.e.} is such that $\Hm \xv = \lambda \sv$ for some $\lambda \in \Fq$. If we find such a vector (and if $\lambda \neq 0$ which will happen with large probability as we will see in what follows) then we get from such an $\xv$ our $\ev$ by taking $\ev = \lambda^{-1} \xv$. The point is that basically all the collision search techniques used for solving the decoding problem get speeded by essentially a factor $q-1$ by identifying all vectors which are proportional. This is particularly helpful in the case where we work 
with big field sizes as is the case for the NIST submission \SDitH \cite{AFGGHJJPRRY23}. 
We will adapt this idea to one of the simplest collision decoding technique, namely Stern's decoding algorithm over $\Fq$ \cite{P10}. 
We call this variant, {\em projective Stern's algorithm} since we work here essentially in the projective space.
We provide here a clean counting of the complexity of this variant of Stern's algorithm  
 in the spirit of \cite[Ch. 6]{P11a}. This precise complexity counting 
 includes the use of the Canteaut-Chabaud technique \cite{CC98} to gain in the 
complexity of Gaussian elimination. This part can not be neglected at all for giving tight security estimates in the case of \SDitH because the list sizes in an optimal Stern algorithm are in this case really small. In \SDitH there is also a variant of the decoding which is considered, which is the $d$-split variant: the support of the error is split into $d$ equal parts and it asked to find
an error $\ev$ of weight $t/d$ on each of the part. We give an adaptation of our projective Stern's algorithm to this case too. 

We will study in detail the impact of this technique to \SDitH both for the initial submission \cite{AFGGHJJPRRY23} and for the recent update that can be found on 
\url{https://sdith.org/}. The initial submission was unfortunately affected by a mistake in the choice of parameters that corresponded to a region where there were several hundred of solutions to the decoding problem whereas the analysis implicitly assumed that there was just one. The security claims made in \cite{AFGGHJJPRRY23} were incorrect because of this.
The new algorithm presented here also reduces the security of this proposal. All in all, this shows that the security of the initial submission \cite{AFGGHJJPRRY23} is below the NIST requirements by 9 to 14 bits depending on the \SDitH variant. Three days after preliminary results of this work were made public \cite{CHT23}, new parameters of \SDitH were released and announced on the NIST forum\footnote{See \url{https://groups.google.com/a/list.nist.gov/g/pqc-forum/c/OOnB655mCN8/m/rL4bPD20AAAJ}}. This new parameter set corrected the initial error in the parameter choice (now the parameters are chosen such that there is typically just one solution to the decoding problem).  The authors took a $4$ bit security margin between the NIST security requirements and the 
estimate for the best attack provided in \url{https://sdith.org/}. We show here that this is still a little bit short of meeting the NIST requirements by roughly one bit. It should be noted that 
contrarily to \cite{AFGGHJJPRRY23} which uses (i) a non tight reduction from standard decoding to $d$-split decoding which gives an overestimate on the attacks, (ii) neglects the cost 
of Gaussian elimination in the attack, our security estimate is based on a precise count of the complexity of the attack which does not neglect the cost of Gaussian elimination. It turns out that the optimal parameters for the projective Stern algorithm are in the regime where the cost of Gaussian elimination is non negligible.

 \section{Preliminaries}\label{sec:notations}

\textbf{Vectors and matrices.} Vectors and matrices are respectively denoted in bold letters 
and bold capital letters such as $\vv$ and $\Mm$. The entry at index $i$ of the vector $\vv$ is denoted by $v_i$. 
$\transp{\Mm}$ stands for the transpose of the matrix $\Mm$. To limit the use of transposition notation as much as possible, we consider in this paper that the vectors are column vectors; so $\transp{\vv}$ represents a row vector.
Let $I$ be a list of indexes. We denote by $\vv_{I}$ the vector $(v_i)_{i \in I}$. In the same way, we denote by $\Mm_{I}$ the submatrix made up of the columns of $\Mm$ which are indexed by $I$.
The notation $\support{\vv}$ stands for the support of $\vv$, that is the set of the non-zero positions of $\vv$.

The double square brackets stand for a set of consecutive integers. For instance, $\IInt{a}{b}$ are the integers between $a$ and $b$.\\

\noindent\textbf{Coding background.}
A linear code $\C$ of length $n$ and dimension $k$ over the field $\Fq$ is a subspace of $\Fq^n$ of dimension $k$. We say that it is an $[n, k]_q$-code. 
It can be defined by a \emph{generator matrix} $\Gm \in \Fq^{k \times n}$ whose rows form a basis of the code:
\begin{equation}
\C \eqdef \left\{\transp{\Gm} \uv\; : \; \uv \in \Fq^k \right\}.
\end{equation}
A \emph{parity-check matrix} for $\C$ is a matrix $\Hm \in \Fq^{(n-k) \times n}$ whose right kernel is $\C$: 
\begin{equation}
\C \eqdef \left\{\cv \in \Fq^n \; : \; \Hm \cv = \Ov \right\}.
\end{equation}

A set of $k$ positions that fully defines a code $\C$ is called an \emph{information set}. In other words, for $I \subseteq \IInt{1}{n}$ such that $\card{I} = k$ and $J \eqdef \IInt{1}{n} \setminus I$, the subset $I$ is an information set if and only if $\Gm_I$ is invertible or equivalently $\Hm_J$ is invertible. In that case, $J$ is called \emph{redundancy set}.

In this paper, we address the decoding problem \ref{pb:SD}. We focus on the case where the decoding distance $t$ is lower than $n - \frac{n}{q}$. We distinguish two particular regimes: when the decoding problem has typically less than one solution and when it has more. For $\Hm \in \Fq^{(n-k) \times n}$ and $\sv \in \Fq^{n-k}$ that are drawn uniformly at random, the greatest distance $t$ for which the decoding problem has less than one solution on expectation is called the \emph{Gilbert-Varshamov distance} and it is denoted by:
\begin{equation}
\dgv(n,k) \eqdef \sup \left( \left\{ t \in \IInt{0}{n - \tfrac{n}{q}} \; : \; \binom{n}{t}(q-1)^t \leq q^{n-k} \right\} \right)
\end{equation}

\textbf{How we measure complexities.} Because one of our goals is to compare our results to those of the \SDitH specifications, we measure the complexities in the same way as they do. In particular, we assume that the additions and multiplications in $\Fq$ are implemented using lookup tables and that these two operations therefore have the same cost. In the \SDitH specifications, this cost is considered as $\log_2(q)$ which is the estimated cost of a memory access. In the following, we count the complexities in number of additions/multiplications and therefore we ignore the factor $\log_2(q)$. However, the results of Section \ref{sec:sdith} are given with this factor. \section{The Stern decoder and Peters' improvements}
\label{sec:stern_peters}
 
In this section, we recall the main results of \cite{S88} and \cite{P10} that are used in \SDitH specifications to solve the decoding problem $\SD(\Hm,\sv,t)$. 
Stern's decoding algorithm is an iterative algorithm that is parametrized by two integers $p$ and $\ell$ to optimize. Each iteration starts by selecting an information set $I \subseteq \IInt{1}{n}$ of size $k$. We denote $J \eqdef \IInt{1}{n}\setminus I$. Then, we search for $\xv \in \Fq^k$ of weight $2p$ such that:
\begin{equation}
\hw{\Pm \xv - \yv} = t - 2p
\end{equation}
where 
\begin{equation}
\Pm \eqdef \Hm_J^{-1} \Hm_I \;\;\;\;\; \text{ and }\;\;\;\;\; \yv \eqdef \Hm_J^{-1} \sv
\end{equation}
We can easily verify that if we find such an $\xv$, then the vector $\ev \in \Fq^n$ defined by $\ev_I \eqdef \xv$ and $\ev_J \eqdef \yv - \Pm \xv$ is a solution to the decoding problem. By making the additional bet that the sought error  $\ev$ is such that $\ev_I$ is of weight $p$ on each half and $\ev_J$ is $\Ov$ on its $\ell$ first positions, 
we can use collision search to find $\ev$ more efficiently. Indeed, using lookup tables, one can find all pairs $(\xv_1, \xv_2) \in \Fq^k \times \Fq^k$ such that $\xv_1$ is zero on its second half (resp. $\xv_2$ is zero on its first half), $\hw{\xv_1} = p$ (resp. $\hw{\xv_2} = p$) and $\Pm \xv_1 - \yv$ and $\Pm \xv_2$ collide on their $\ell$ first positions. Thus, for each of these collisions, the vector $\ev$ such that $\ev_I \eqdef \xv_1 + \xv_2$ and $\ev_J \eqdef \yv - \Pm (\xv_1+\xv_2)$ is a potential solution to the \SD problem because it has syndrome $\sv$ and it is of particular low weight on at least $k+\ell$ positions. Finally, Algorithm \ref{algo:Stern} summarizes the Stern decoder.\\

\SetKwInput{KwParams}{Parameters}
\SetKwFor{RepTimes}{repeat}{times}{end}
\SetKwFor{RepAsNecessary}{repeat as many times as necessary}{}{end}

\begin{algorithm}[h]
    \DontPrintSemicolon 

    \KwIn{$\Hm \in \Fq^{(n-k) \times n}$, $\sv \in \Fq^{n-k}$ and $t \in \IInt{0}{n}$.}
    \KwParams{$p \in \IInt{0}{\frac{\min(t,k)}{2}}$ and $\ell \in \IInt{0}{n-k - t + 2p}$.}
    \KwOut{$\ev \in \Fq^n$ such that $\Hm \ev = \sv$ and $\hw{\ev}=t$.}
	\vspace{0.2cm}
	 \RepAsNecessary{}{
	 	draw $I \subseteq \IInt{1}{n}$ of size $k$ uniformly at random \label{algomark:info}\;
	 	$J \gets \IInt{1}{n} \setminus I$ \;
	 	$\Pm \gets \Hm_J^{-1}\Hm_I$ \label{algomark:Gauss}\tcc*{\small if $\Hm_J$ is not invertible, go back to step \ref{algomark:info}}
	 	$\yv \gets \Hm_J^{-1} \sv$ \;
	 	$\Rm \gets$ the $\ell$ first rows of $\Pm$ \;
	 	$\zv \gets$ the $\ell$ first positions of $\sv$ \;
	 	$\sL_1 \gets \left\lbrace \Rm \xv_1 - \zv \; : \; \xv_1 \in \Fq^{\lfloor k/2 \rfloor} \times 0^{k-\lfloor k/2 \rfloor} \text{ and } \hw{\xv_1} = p \right\rbrace$ \;
		$\sL_2 \gets \left\lbrace \Rm \xv_2 \; : \; \xv_2 \in 0^{\lfloor k/2 \rfloor} \times \Fq^{k-\lfloor k/2 \rfloor} \text{ and } \hw{\xv_2} = p \right\rbrace$ \;
		
		\ForAll{$\left( \Rm \xv_1 - \zv, \Rm \xv_2 \right) \in \sL_1 \times \sL_2$ such that $\Rm \xv_1 - \zv = \Rm \xv_2$  \label{algomarker:forloop}}{
			\If{$\hw{\Pm (\xv_1 - \xv_2) - \yv} = t - 2p$ \label{algomarker:check}}{
			\Return $\ev$ such that $\ev_I \eqdef \xv_1 - \xv_2$ and $\ev_J \eqdef \yv - \Pm (\xv_1 - \xv_2)$ \;
			}
		}
	 }

    \caption{Stern's algorithm to solve $\SD(\Hm,\sv,t)$\label{algo:Stern}}

\end{algorithm}

Note that if a particular error vector $\ev$ has the good weight distribution -- that is $\ev_I$ is of weight $p$ on each half, $\ev_J$ is $\Ov$ on its $\ell$ first positions and of weight $t - 2p$ on its $n-k-\ell$ other positions -- then Stern's algorithm will find $\ev$. So the probability to find a particular solution is
\begin{equation}
p_{\mathrm{part}} = \dfrac{\binom{\lfloor k/2 \rfloor}{p}\binom{k - \lfloor k/2 \rfloor}{p} \binom{n - k - \ell}{t - 2p}}{\binom{n}{t}}
\end{equation}
Moreover, in the case where $\sv$ has been produced as the syndrome of an error $\widetilde{\ev}$ of weight $t$ -- that means $\sv$ has been drawn uniformly at random in $\{ \Hm \widetilde{\ev} \; : \; \hw{\widetilde{\ev}} = t \}$ -- then, the expected number of solutions to the decoding problem we address is

\begin{eqnarray}
N_{\mathrm{sol}} & \eqdef & \expect{\Hm}{\card{\left\{ \ev \in \Fq^n \; : \; \hw{\ev}=t \text{ and } \Hm \ev = \Hm \widetilde{\ev} \right\}}} \\[0.2cm]
 & = & 1 + \sum_{\substack{\ev \in \Fq^n \setminus \{\widetilde{\ev}\} \\ \hw{\ev}=t} } \prob{\Hm}{\Hm \ev = \Hm \widetilde{\ev} } \\[0.2cm]
 & = & 1 + \frac{\binom{n}{t}(q-1)^{t} - 1}{q^{n-k}}\label{eq:Nsol}
\end{eqnarray}

Thus, the success probability of one iteration of Stern's algorithm is
\begin{equation}
\Psucc  =  1 - \left(1 - p_{\mathrm{part}} \right)^{N_{\mathrm{sol}}}
\end{equation}
So on average over the choice of $\Hm$, we need to repeat Stern's procedure $\frac{1}{\Psucc}$ times before finding a solution to the decoding problem.
To determine the complexity of Stern's algorithm, we still have to measure the time complexity of one iteration. Using some of the tricks proposed in \cite{P10}, the designers of \SDitH claim to be able to perform each iteration of Stern with a running time
\begin{equation}
\begin{array}{lcl}
T_{\mathrm{iter}} & = & \tfrac{1}{2} (n-k)^2(n+k) \\[0.2cm]
& & \;\;\;\; + \ell \left( \tfrac{k}{2} - p + 1 + \left( \binom{\lfloor k/2 \rfloor}{p} + \binom{k - \lfloor k/2 \rfloor}{p} \right) (q-1)^p\right) \\[0.2cm]
& & \;\;\;\;+ \tfrac{q}{q-1}(t-2p+1) 2p\left( 1 + \tfrac{q-2}{q-1}\right) \dfrac{\binom{\lfloor k/2 \rfloor}{p} \binom{k - \lfloor k/2 \rfloor}{p} (q-1)^{2p}}{q^{\ell}}
\end{array}
\end{equation}

 \section{Reducing the decoding problem to the low weight codeword search}\label{sec:reduction}

Working in projective spaces is only interesting if the syndrome is zero. In that case, we are actually looking for a low weight codeword instead of an error vector. 
This can be readily achieved by a well known reduction from decoding  in an $[n,k]_q$ linear code  to a low weight codeword search in an $[n,k+1]_q$ linear code 
that we now recall.
Let $\Hm$ be a parity check matrix of a code $\C$. Without loss of generality, we can consider that $\Hm$ is in systematic form\footnote{The first operation of Stern's algorithm precisely consists in putting the parity-check matrix in systematic form (up to a permutation). And therefore making this assumption does not induce any additional cost.}:
\begin{equation}
\label{eq:systematic}
\Hm \eqdef \left[ \Am | \Id_{n - k} \right] \;\;\;\;\; \text{ where } \Am \in \Fq^{(n-k) \times k}
\end{equation}
To solve the decoding problem $\SD(\Hm, \sv, t)$, one can find a low weight codeword in the new code 
\begin{equation}
\C' \eqdef \langle\C, \zv \rangle \eqdef \left\{ \cv + \alpha \zv \; : \; \cv \in \C \text{ and } \alpha \in \Fq\right\}
\end{equation}
where $\zv \in \Fq^n$ is any solution of the equation $\Hm \zv = \sv$ (without any weight constraint on $\zv$). Because $\Hm$ is in systematic form, we can  take $\transp{\zv}   \eqdef (\transp{\Ov}, \transp{\sv}) \in \Fq^n$. Then a generator matrix of $\C'$ is
\begin{equation}
\Gm' \eqdef \left[ 
\begin{matrix}
\Id_{k} & -\transp{\Am} \\
\transp{\Ov} & \transp{\sv}
\end{matrix}
\right]
\end{equation}

By only one step of a Gaussian elimination (one column to eliminate), we can find a parity-check matrix $\Hm'$ of the augmented code $\C'$.\\

By looking for a low weight codeword in $\C'$ -- \ie a vector $\ev$ such that $\Hm' \ev = \Ov$ --, we actually find a low weight error $\ev$ of $\C$ that has syndrome $\Hm \ev = \alpha \sv$ where $\alpha$ can be any scalar in $\Fq$. There are two possible situations: either $\alpha = 0$ or $\alpha \neq 0$. If $\alpha = 0$ then we have actually found a codeword in $\C$ instead of an error vector (we want to avoid this situation). On the contrary, if $\alpha \neq 0$ then a solution to our original decoding problem is simply $\alpha^{-1} \ev$. We now claim that the probability to get $\alpha = 0$ is lower than $\frac{1}{q}$:

\begin{theorem}\label{th:reduction}
Let a code $\C$ be the right kernel of a parity-check matrix $\Hm \in \Fq^{(n-k)\times n}$ and let $\sv \in \left\lbrace \Hm \widetilde{\ev} \; : \; \hw{\widetilde{\ev}}=t \right\rbrace$ for $t \in \IInt{0}{n}$.
Let $\C' \eqdef \langle\C, \zv \rangle$ be the code generated by the codewords in $\C$ and any word $\zv \in \Fq^n$ such that $\Hm \zv = \sv$. We denote by $\Hm'$ a parity-check matrix of this augmented code.
Then we can solve the decoding problem $\SD(\Hm, \sv, t)$ by solving, on average over $\Hm$, at most $\frac{q}{q-1}$ low weight codeword searches $\SD(\Hm', \Ov, t)$.
\end{theorem}

\begin{proof}
First, by construction of $\sv$ and $\zv$, we know there exists a codeword $\cv \in \C$ and an error vector $\widetilde{\ev} \in \Fq^n$ of weight $t$ such that $\zv = \cv + \widetilde{\ev}$. So we have $\langle\C, \zv \rangle = \langle\C, \widetilde{\ev} \rangle$ and $\sv = \Hm \widetilde{\ev}$.

Let $\ev$ be any solution of the low weight codeword search problem $\SD(\Hm', \Ov, t)$. Then, as said before, $\ev$ is an error vector of weight $t$ such that $\Hm \ev = \alpha \sv$ for a scalar $\alpha \in \Fq$, and $\ev$ induces a solution of the original decoding problem if and only if $\alpha \neq 0$. On average over $\Hm$ and for $\ev$ drawn uniformly at random in the solutions of $\SD(\Hm', \Ov, t)$, the probability that $\alpha = 0$ is
$$
\prob{\Hm, \ev}{\alpha = 0} = \frac{\expect{\Hm}{\card{\C(t)}}}{\expect{\Hm}{\card{\C'(t)}}}
$$
where 
$$
\begin{array}{lcl}
\C(t) & \eqdef & \left\{ \ev \in \C \; : \; \hw{\ev} = t \right\} \\
\C'(t) & \eqdef & \left\{ \ev \in \C' \; : \; \hw{\ev} = t \right\}
\end{array}
$$
We already know that 
$$
\expect{\Hm}{\card{\C(t)}} = \frac{\binom{n}{t}(q-1)^t}{q^{n-k}}.
$$
Let us count $\C'(t)$. We remark that $\C'$ is the disjoint union of the cosets $\C + \alpha \widetilde{\ev}$ for all the $\alpha \in \Fq$. So 
$$
\expect{\Hm}{\card{\C'(t)}} = \sum_{\alpha \in \Fq} \expect{\Hm}{\card{\C'_{\alpha}(t)}}
$$
where
$$
\C'_{\alpha}(t) \eqdef \left\{ \ev + \alpha \widetilde{\ev} \; : \; \ev \in \C \text{ and }\hw{\ev + \alpha \widetilde{\ev}} = t \right\}.
$$
On one hand, we have that $\C'_{0}(t) = \C(t)$. On another hand, for all non-zero $\alpha, \alpha' \in \F_q^*$, the map
$\xv \longmapsto \alpha^{-1} \alpha' \xv$ is a bijection from $\C'_{\alpha}(t)$ to $\C'_{\alpha'}(t)$; so for all $\alpha \in \F_q^*$, $
\card{\C'_{\alpha}(t)}=\card{\C'_{1}(t)}.
$
Thus, we have
$$
\expect{\Hm}{\card{\C'(t)}} = \expect{\Hm}{\card{\C(t)}} + (q-1) \expect{\Hm}{\card{\C'_{1}(t)}}
$$

By doing a calculation similar to that of the Equation \eqref{eq:Nsol}, we can show that
$$
\begin{array}{lcl}
\expect{\Hm}{\card{\C'_{1}(t)}} & = & 1 + \frac{\binom{n}{t}(q-1)^t - 1}{q^{n-k}} \\[0.2cm]
 & = & 1 - \frac{1}{q^{n-k}} + \expect{\Hm}{\card{\C(t)}}.
\end{array}
$$
and so
$$
\expect{\Hm}{\card{\C'(t)}} = q \expect{\Hm}{\card{\C(t)}} + (q-1) \left( 1 - \tfrac{1}{q^{n-k}}\right).
$$
Finally, the probability that the reduction succeeds is
$$
1 - \prob{\Hm, \ev}{\alpha = 0}  =  1 - \frac{1}{q + \tfrac{(q-1)(1 - 1/q^{n-k})}{\expect{\Hm}{\card{\C(t)}}}}  \geq  \frac{q-1}{q}.
$$
\end{proof}

 \section{Stern's algorithm in projective space}
\label{sec:projection}

In Section \ref{sec:reduction}, we gave a reduction of  decoding  to low weight codeword searching. In this section, we address the second problem, that is given a parity-check matrix $\Hm' \in \Fq^{(n-k-1) \times n}$ of $\C'$, we want to find $\ev \in \Fq^{n}$ such that $\hw{\ev} = t$ and $\Hm' \ev = \Ov$.

We have to be careful about the distribution of $\Hm'$ which has not been drawn uniformly at random in $\Fq^{(n-k-1) \times n}$. Indeed, for an error $\widetilde{\ev} \in \Fq^n$ of weight $t$, $\Hm'$ has been drawn such that $\widetilde{\ev}$ is a codeword in $\C'$, so $\Hm'$ verifies $\Hm' \widetilde{\ev} = \Ov$.\\
Essentially, our method consists in running a Stern procedure in the projective space $\projspace{\Fq^n}$. In particular, we show that Peters' improvements of Stern \cite[Ch. 6]{P11a} are still applicable in the projective space.

\subsection{The algorithm}

When the syndrome is zero, Stern's algorithm essentially consists in finding pairs $(\xv_1, \xv_2) \in \Fq^{k+1} \times \Fq^{k+1}$ such that $\xv_1$ (resp. $\xv_2$) is of weight $p$ on its first $\lfloor \frac{k+1}{2} \rfloor$ (resp. last $k+1 - \lfloor \frac{k+1}{2} \rfloor$) positions, zero elsewhere and 
\begin{equation}
\Rm' \xv_1 = \Rm' \xv_2
\end{equation}
where $\Rm'$ are the $\ell$ first rows of $\Pm' \eqdef \Hm'^{-1}_J \Hm'_I$. Each pair $(\xv_1, \xv_2)$ that collides gives a candidate codeword $\ev$ defined by
\begin{equation}
\ev_I = \xv_1 - \xv_2 \;\;\;\;\; \text{ and }\;\;\;\;\; \ev_J = \Pm' (\xv_1 - \xv_2).
\end{equation}

One can remark that if the pair $(\xv_1, \xv_2)$ is a solution to the collision search, then for all $\alpha \in \Fq^*$, $\alpha \xv_1$ and $\alpha \xv_2$ also collide.

\begin{remark}
Note that this trick is specific to the fact that the syndrome is zero. If the syndrome is non-zero, then given a pair $(\xv_1,\xv_2)$ that is such that $\Rm' \xv_1 - \yv' = \Rm' \xv_2$, we can no longer guarantee that for any non-zero $\alpha$, we still have $\alpha \Rm' \xv_1 - \yv' = \alpha \Rm' \xv_2$. That is the reason why the reduction in Section \ref{sec:reduction} is essential.
\end{remark}

Moreover, $\alpha \xv_1$ and $\alpha \xv_2$ respectively share the same support as $\xv_1$ and $\xv_2$ so the Stern procedure enumerates all the collinear equivalents of $\xv_1$ and $\xv_2$ and consequently, it explores all the candidate codewords that are collinear to $\ev$. However, we only need one of them. Indeed, if $\ev$ is in $\C'$ but not in $\C$ -- that is $\Hm' \ev = \Ov$ but $\Hm \ev \neq \Ov$ -- then there is a unique $\alpha \in \Fq^*$ such that the syndrome $\alpha \Hm \ev$ is exactly $\sv$ and not a multiple of it. The solution $\alpha \ev$ can be found from any vector that is collinear to $\ev$ and so, we only need to find one of them.\\

From the discussion above, we remark that when the syndrome is zero, Stern's algorithm can be run in the projective space. For a space $\mathcal{E}$ over $\Fq$, the projective space $\projspace{\mathcal{E}}$ is the quotient set of $\mathcal{E}$ by the \emph{equivalence relation} $\sim$:
\begin{equation}
\forall \xv,\yv \in \mathcal{E}, \;\;\; \xv \sim \yv \; \Longleftrightarrow \; \exists \alpha \in \Fq^*, \; \xv = \alpha \yv.
\end{equation}

The \emph{equivalence class} of a vector $\xv \in \mathcal{E}$ is denoted by:
\begin{equation}
\class{\xv} \eqdef \left\{ \yv \in \mathcal{E} \; : \; \yv \sim \xv\right\}.
\end{equation}
And so 
\begin{equation}
\projspace{\mathcal{E}} \eqdef \left\{ \class{\xv} \; : \; \xv \in \mathcal{E} \right\}.
\end{equation}

Now, if $\xv_1, \xv_2 \in \Fq^{k+1}$ are such that $\Rm' \xv_1 = \Rm' \xv_2$ then we also have $\Rm' \class{\xv_1} = \Rm' \class{\xv_2}$ where $\class{\xv_1}$ and $\class{\xv_2}$ live in $\projspace{\Fq^{k+1}}$. 
But we have to note that if $\Rm' \class{\xv_1} =  \Rm' \class{\xv_2}$ then we do not necessarily have $\Rm' \xv_1 =  \Rm' \xv_2$. So we need to choose representatives $\rep{\xv_1} \in \class{\xv_1}$ and $\rep{\xv_2} \in \class{\xv_2}$ which guarantee
\begin{equation}
\Rm' \rep{\xv_1} =  \Rm' \rep{\xv_2}
\end{equation}

To do that, we distinguish a particular class representative:
\begin{definition}[Particular class representative]
\label{def:repr}
Let $\Rm' \in \Fq^{\ell \times (k+1)}$. For all $\class{\xv} \in \projspace{\Fq^{k+1}}$, if $\Rm' \xv = \Ov$ then the vector $\rep{\xv}$ is any representative of $\class{\xv}$, otherwise it is the unique representative of $\class{\xv}$ such that the first non-zero symbol of $\Rm' \rep{\xv}$ is $1$. 
\end{definition}

\begin{lemma}
\label{lemma:repr}
For any $\xv_1, \xv_2 \in \Fq^{k+1}$, 
\begin{equation}
\Rm' \class{\xv_1} =  \Rm' \class{\xv_2} \;\;\; \Longleftrightarrow \;\;\; \Rm' \rep{\xv_1} =  \Rm' \rep{\xv_2}.
\end{equation}
\end{lemma}

\begin{proof} If $\Rm' \xv_1 =  \Rm' \xv_2 = \Ov$ then the proof is trivial. Otherwise, for either $i=1$ or $2$, $\Rm' \class{\xv_i}$ is made of all the vectors that are collinear to $\Rm' \xv_i$. Thus, all the elements in $\Rm' \class{\xv_i}$ share the same support, in particular they have the same first non-zero position, and there is a unique vector in $\Rm' \class{\xv_i}$ for which this first non-zero position contains a $1$. So, we first notice that $\rep{\xv_1}$ and $\rep{\xv_2}$ exist and they are unique. 

Assume $\Rm' \class{\xv_1} =  \Rm' \class{\xv_2}$. That means $\Rm' \rep{\xv_1} \in  \Rm' \class{\xv_2}$. On another hand, the first non-zero symbol in $\Rm' \rep{\xv_1}$ is a one and the only element of this kind in $\Rm' \class{\xv_2}$ is $\Rm' \rep{\xv_2}$, so we necessarily have $\Rm' \rep{\xv_1} = \Rm' \rep{\xv_2}$.

Conversely, $\Rm' \rep{\xv_1} =  \Rm' \rep{\xv_2} \; \Rightarrow \; \class{\Rm' \rep{\xv_1}} =  \class{\Rm' \rep{\xv_2}} \; \Rightarrow \; \Rm' \class{\xv_1} =  \Rm' \class{\xv_2}$.

\end{proof}

Now we are ready to describe our adaptation of Stern's algorithm to the projective space. Algorithm \ref{algo:SternProj} gives the pseudo code of the method. Note that unlike Algorithm \ref{algo:Stern}, here the syndrome is zero and $\rep{\xv_1}, \rep{\xv_2} \in \Fq^{k+1}$ are some particular representatives of $\class{\xv_1}, \class{\xv_2} \in \projspace{\Fq^{k+1}}$. Moreover, we must treat differently the case where $\Rm' \rep{\xv_1} = \Rm' \rep{\xv_2} = \Ov$ because this case generates $q-1$ collisions that are not collinear with each other.
Lemma \ref{lemma:repr} guarantees that a collision in projective space is still a collision when using the good representative so that guarantees the correctness of the algorithm.

\begin{algorithm}
    \DontPrintSemicolon 

    \KwIn{$\Hm' \in \Fq^{(n-k-1) \times n}$ and $t \in \IInt{0}{n}$.}
    \KwParams{$p \in \IInt{0}{\frac{\min(t,k+1)}{2}}$ and $\ell \in \IInt{0}{n - k - 1 - t + 2p}$.}
    \KwOut{$\ev \in \Fq^n$ such that $\Hm' \ev = \Ov$ and $\hw{\ev}=t$.}
	\vspace{0.2cm}
	 \RepAsNecessary{}{
	 	draw $I \subseteq \IInt{1}{n}$ of size $k+1$ uniformly at random \label{algomark:infoProj}\;
	 	$J \gets \IInt{1}{n} \setminus I$ \;
	 	$\Pm' \gets \Hm'^{-1}_J\Hm'_I$ \label{algomark:GaussProj}\tcc*{\small if $\Hm'_J$ is not invertible, go back to step \ref{algomark:info}}
	 	$\Rm' \gets$ the $\ell$ first rows of $\Pm'$ \label{algomark:ellProj} \;
	 	$\sL'_1 \gets \left\lbrace \Rm' \rep{\xv_1} \; : \; \class{\xv_1} \in \projspace{\left(\Fq^{\lfloor \frac{k+1}{2}\rfloor} \times 0^{k+1 - \lfloor \frac{k+1}{2}\rfloor} \right)} \text{ and } \hw{\xv_1} = p \right\rbrace$ \;
		$\sL'_2 \gets \left\lbrace  \Rm' \rep{\xv_2} \; : \; \class{\xv_2} \in \projspace{\left(0^{\lfloor \frac{k+1}{2}\rfloor} \times \Fq^{k+1 - \lfloor \frac{k+1}{2}\rfloor}\right)} \text{ and } \hw{\xv_2} = p \right\rbrace$ \;
		\ForAll{$\left(\Rm' \rep{\xv_1}, \Rm' \rep{\xv_2} \right) \in \sL'_1 \times \sL'_2$ such that $\Rm' \rep{\xv_1}= \Rm' \rep{\xv_2}$\label{algomark:bouclecheckproj}}{
			\uIf{$\Rm' \rep{\xv_1} = \Ov \textbf{\em \ and }\exists \alpha \in \Fq^*, \hw{ \Pm' (\alpha \rep{\xv_1} - \rep{\xv_2})} = t - 2p$}{
						\Return $\ev$ such that $\ev_I = \alpha\rep{\xv_1} - \rep{\xv_2}$ and $\transp{\ev_J} = \Pm' (\alpha\rep{\xv_1} - \rep{\xv_2})$ \;
			}\ElseIf{$\hw{ \Pm' (\rep{\xv_1} - \rep{\xv_2})} = t - 2p$ \label{algomarker:checkProj}}{
				\Return $\ev$ such that $\ev_I = \rep{\xv_1} - \rep{\xv_2}$ and $\transp{\ev_J} = \Pm' (\rep{\xv_1} - \rep{\xv_2})$ \label{algomark:returnProj} \;
			}
		}
	 }

    \caption{Projective Stern's algorithm to solve $\SD(\Hm',\Ov,t)$\label{algo:SternProj}}
\end{algorithm}

\subsection{Reducing the cost of Gaussian elimination}
\label{sec:Gauss}

For large $q$ and $p$, 
the Gaussian elimination step is negligible and so, we can afford to perform it on $n-k-1$ columns drawn independently at each iteration. Thus, Gaussian elimination needs $(n-k-1)^2 (n + k + 2)$ operations. However, in the context of \SDitH, we are far away from this regime and the Gaussian elimination is actually one of the most expensive operation we have to perform. In this sub-section, we present two modifications of our original projective Stern algorithm that allow to reduce the impact of the Gaussian elimination step. Those tricks are inspired by \cite{P10} and \cite{BLP08} and have been adapted to our situation.\\

\textbf{Factorizing the Gaussian elimination step.} An iteration of Algorithm \ref{algo:SternProj} begins with selecting an information set $I$ and a window of size $\ell$. Let denote by $I_1$ (resp. $I_2$) the first half of $I$ (resp. the second half of $I$) and $J_{\ell}$ the $\ell$ first positions of $J \eqdef \IInt{1}{n}\setminus I$. The iteration succeeds in finding the particular error vector $\ev$ of weight $t$ if it verifies
\begin{equation}
\hw{\ev_{I_1}} = \hw{\ev_{I_2}} = p \;\;\;\;\; \text{ and }\;\;\;\;\; \hw{\ev_{J_{\ell}}} = 0.
\end{equation}

To save some Gaussian elimination steps, we can test several partitions $(I_1, I_2, J_{\ell})$ for one given information set. In other words, the main loop in Algorithm \ref{algo:SternProj} can be divided into an \emph{outer loop} and an \emph{inner loop}. The \emph{outer loop} consists in selecting an information set $I$ and performing a Gaussian elimination on it (steps \ref{algomark:infoProj}-\ref{algomark:GaussProj}). The \emph{inner loop} starts by partitioning $I$ into $(I_1, I_2)$ and selecting a window $J_{\ell} \subset J$ of size $\ell$, then it performs the steps \ref{algomark:ellProj}-\ref{algomark:returnProj} with
\begin{eqnarray}
\Rm' & \eqdef & \text{ The rows of $\Pm'$ indexed by $J_{\ell}$} \\
\sL'_1 & \eqdef & \left\lbrace \Rm' \rep{\xv_1} \; : \; \class{\xv_1} \in \projspace{\Fq^{k+1}} \text{ and } \support{\xv_1} \subseteq I_1 \text{ and } \hw{\xv_1} = p \right\rbrace \\
\sL'_2 & \eqdef & \left\lbrace  \Rm' \rep{\xv_2} \; : \; \class{\xv_2} \in \projspace{\Fq^{k+1}} \text{ and } \support{\xv_2} \subseteq I_2 \text{ and } \hw{\xv_2} = p \right\rbrace
\end{eqnarray}

For a given information set $I$, we choose the partition $(I_1, I_2, J_{\ell})$ uniformly at random and independently from one iteration to another. Assuming we are looking for a $t$-weight codeword $\ev \in \C'$ that verifies 
$\hw{\ev_I} = 2p$, then the success probability of finding this particular codeword during an iteration of the inner loop is
\begin{equation}
q_{\mathrm{in}} = \dfrac{\binom{\lfloor \frac{k+1}{2}\rfloor}{p} \binom{k+1 - \lfloor \frac{k+1}{2}\rfloor}{p} \binom{n - k - 1 - \ell}{t - 2p} }{\binom{k + 1}{2p} \binom{n - k - 1}{t - 2p}}. \label{eq:qin}
\end{equation}

So the number of trials needed to get $\ev$ follows a geometric distribution of parameter $q_{\mathrm{in}}$ and so, by iterating $N_{\mathrm{in}}^{\mathrm{tmp}}$ times the inner loop, we will find $\ev$ with probability
\begin{equation}
p_{\mathrm{in}} \eqdef 1 - (1 - q_{\mathrm{in}})^{N_{\mathrm{in}}^{\mathrm{tmp}}}
\end{equation}
Note that taking
\begin{equation}
N_{\mathrm{in}}^{\mathrm{tmp}}\eqdef \frac{1}{q_{\mathrm{in}}}
\end{equation}
allows to achieve a success probability $p_{\mathrm{in}}$ for the inner loop that is exponentially close to $1$.\\

\textbf{Reusing pivots in the Gaussian elimination.} In \cite{CC98}, Canteaut and Chabaud propose to simplify the Gaussian elimination step by changing only one index in the information set $I$. Thus, only one pivot is necessary from one iteration of the outer loop to another. This idea is generalized in \cite{BLP08} where this time, the number of columns to eliminate from one iteration to another can be greater than $1$. By doing this, we reduce the cost of Gaussian elimination but we also induce some dependencies between the selected information sets that impact the number of iterations of the outer loop that is needed.

To estimate the impact of the technique described above, we lean on the analysis in \cite{BLP08,P10}. We first introduce the parameter $c$ which represents the number of columns to eliminate in each iteration\footnote{In \cite{BLP08}, another parameter $r$ is introduced but its interest is only  for small field.}. Thus, the cost of Gaussian elimination per iteration of the outer loop is
\begin{eqnarray}
T_{\mathrm{Gauss}} & = & 2 \sum_{i=1}^{c} (n - k - 1)(k + 1 + i) \\
 & = & c(n-k-1)(2k+c+3) \label{eq:TGaussProj}
\end{eqnarray}

Note that if a $t$-weight codeword $\ev \in \C'$ is such that $\hw{\ev_I} = 2p$, then the corresponding iteration of the outer loop will find a representative of $\class{\ev}$ with probability $p_{\mathrm{in}}$. So we need to count the average number of iterations of the outer loop that we need for having this particular weight distribution. However, there are some dependencies between the iterations that must be taken into account. Indeed, we do not draw the $k+1$ positions of the information set independently from one iteration to another ($k+1-c$ positions are kept).

The situation can be modeled by a $(t+2)$-state absorbing Markov chain. Given a $t$-weight codeword $\ev \in \C'$, let $X_i$ be the random variable that represents the weight of $\ev_I$ at iteration $i \in \NN$ of the outer loop or ``$\done$'' if the previous iteration succeeds. For the first iteration, the information set $I$ is chosen uniformly at random as a subset of $\IInt{1}{n}$ of size $k$. So the distribution of $X_0$ is given by

\begin{equation}
\forall v \in \IInt{0}{t}, \;\;\; \prob{}{X_0 = v} = \frac{\binom{k+1}{v} \binom{n - k - 1}{t - v}}{\binom{n}{t}} \;\;\;\;\; \text{ and } \;\;\;\;\; \prob{}{X_0 = \done} = 0.
\end{equation}

Let $\mat{\Pi}$ be the transition matrix of the Markov chain. It is defined as the following stochastic matrix:
\begin{equation}
\forall (u,v) \in \{\done, 0,\cdots, t \}^2, \;\; \mat{\Pi}[u,v] \eqdef \prob{}{X_{i+1} = v \; \vert \; X_i = u}
\end{equation}

The state $\done$ is the absorbing state, that means when we are in this state, we cannot get out anymore. So we have
\begin{equation}
\forall v \in \IInt{0}{t}, \;\; \mat{\Pi}[\done,v] = 0 \;\;\;\;\; \text{ and } \;\;\;\;\; \mat{\Pi}[\done,\done] = 1.
\end{equation}

From an iteration to another, the information set $I$ is updated by swapping $c$ indexes drawn uniformly at random in $I$ with $c$ indexes drawn uniformly at random in $J$. So an iteration moves from state $u$ to state $v$ with probability
\begin{equation}
\label{eq:transition}
\mat{\Pi}[u,v] =   \sum_{j}\dfrac{\binom{u}{j} \binom{k + 1 - u}{c - j} \binom{t - u}{v - u + j} \binom{n - k - 1 - t + u}{c - v + u - j}}{\binom{k+1}{c} \binom{n-k-1}{c}}
\end{equation}
except for $u = 2p$ because then the algorithm succeeds with probability:
\begin{equation}
\mat{\Pi}[2p,\done] =  p_{\mathrm{in}}.
\end{equation}
So for all $v \in \IInt{0}{t}$:
\begin{equation}
\mat{\Pi}[2p,v] = \left( 1 - p_{\mathrm{in}} \right) \cdot \sum_{j} \dfrac{ \binom{2p}{j} \binom{k + 1 - 2p}{c - j} \binom{t - 2p}{v - 2p + j} \binom{n - k - 1 - t + 2p}{c - v + 2p - j}}{\binom{k+1}{c} \binom{n-k-1}{c}}.
\end{equation}

Finally, to determine the number of iterations needed to get the first success, one only has to compute the fundamental matrix associated to $\mat{\Pi}$:
\begin{equation}
\mat{F} \eqdef \left( \Id_{t+1} - \mat{\Pi'} \right)^{-1}
\end{equation}
where $\Id_{t+1}$ is the identity matrix of size $t+1$ and $\mat{\Pi'}$ is the $(t+1) \times (t+1)$ sub-matrix of $\mat{\Pi}$ such that
\begin{equation}
\forall (u,v) \in \{0,\cdots,t\}^2, \;\; \mat{\Pi'}[u,v] \eqdef \mat{\Pi}[u,v].
\end{equation} 
Then, the average number of iterations of the outer loop needed to find a representative of $\class{\ev}$ is
\begin{equation}
N_{\mathrm{out}}^{\mathrm{tmp}} = \sum_{u = 0}^{t} \sum_{v = 0}^{t} \prob{}{X_0 = v}  \mat{F}[u,v]. \label{eq:Nout_tmp}
\end{equation}

\textbf{Finding one solution from many.} With $N_{\mathrm{out}}^{\mathrm{tmp}} \cdot N_{\mathrm{in}}^{\mathrm{tmp}}$ repetitions of the inner loop, we are able to find one particular $t$-weight projective codeword $\class{\ev} \in \projspace{\C'}$. But there is potentially more than one such projective codeword since this number is
\begin{eqnarray}
N_{\mathrm{sol}} & = & \expect{\Hm'}{\card{\left\{ \class{\ev} \in \projspace{\Fq^n} \; : \; \hw{\ev}=t \text{ and } \Hm' \class{\ev} = \class{\Ov} \right\}}} \\
 & = & 1 + \frac{\binom{n}{t}(q-1)^{t-1} - 1}{q^{n-k-1}} \label{eq:Nsol_proj}
\end{eqnarray}
Because we want to find only one solution from the $N_{\mathrm{sol}}$ ones, we actually need to approximately iterate the outer loop $N_{\mathrm{out}}$ times and for each iteration of the outer loop, we iterate the inner loop $N_{\mathrm{in}}$ times where
\begin{eqnarray}
N_{\mathrm{out}} & \eqdef & \max \left(1, \frac{N_{\mathrm{out}}^{\mathrm{tmp}}}{N_{\mathrm{sol}}} \right) \label{eq:Nout}\\
N_{\mathrm{in}} & \eqdef & \max \left(1, N_{\mathrm{in}}^{\mathrm{tmp}} \cdot \min \left( 1 , \frac{N_{\mathrm{out}}^{\mathrm{tmp}}}{N_{\mathrm{sol}}}\right) \right). \label{eq:Nin}
\end{eqnarray}

\subsection{Complexity of our projective Stern decoding}

Finally, considering the modifications of the previous sub-section and using the implementation tricks of Peters \cite{P10}, we are able to state the following Theorem \ref{th:complexitySternProj} that gives the complexity of Stern's algorithm in projective space.

\begin{theorem}\label{th:complexitySternProj}
Let $\widetilde{\ev} \in \Fq^{n}$ be such that $\hw{\ev} = t$. On average over the choice of $\Hm' \in \Fq^{(n-k-1)\times n}$ that is such that $\Hm' \widetilde{\ev} = \Ov$, we can solve the low weight codeword search problem $\SD(\Hm', \Ov, t)$ with a running time of order
\begin{equation}
T_{\mathrm{Stern-proj}} = N_{\mathrm{out}} \Big( T_{\mathrm{Gauss}} +  N_{\mathrm{in}} \big( T_{\mathrm{lists}} + T_{\mathrm{check}} \big) \Big)
\end{equation}
where $N_{\mathrm{out}}$, $N_{\mathrm{in}}$ and $T_{\mathrm{Gauss}}$ are given by Equations \eqref{eq:Nout}, \eqref{eq:Nin} and \eqref{eq:TGaussProj},  and
\begin{eqnarray}
L_1 & = & \binom{\lfloor\frac{k+1}{2}\rfloor}{p}(q-1)^{p-1} \label{eq:list1_proj}\\
L_2 & = & \binom{k + 1 - \lfloor\frac{k+1}{2}\rfloor}{p}(q-1)^{p-1}\label{eq:list2_proj}\\
T_{\mathrm{lists}} & = & \ell \left( k + 2p - 1 + 2 \left( L_1 + L_2 \right) \right)  \\
N_{\mathrm{collisions}} & = & \dfrac{(q-1) L_1 L_2}{q^{\ell}} \\
T_{\mathrm{check}} & = & \left(2p + \tfrac{q}{q-1}(t-2p+1) 2p\left( 1 + \tfrac{q-2}{q-1}\right)\right) \cdot N_{\mathrm{collision}} 
\end{eqnarray}
\end{theorem}

\begin{proof}
According to the Sub-section \ref{sec:Gauss}, we have to iterate $N_{\mathrm{out}}$ times the outer loop which consists in a Gaussian elimination over $c$ columns and $N_{\mathrm{in}}$ iterations of the inner loop. All that remains is to determine the cost of one iteration of the inner loop. To perform this iteration optimally, we will use Peters' implementation tricks \cite{P10}.

To build the lists $\sL_1$ and $\sL_2$, we need to define another representative of an equivalence class in $\projspace{\Fq^{k+1}}$. Let $\xv \in \Fq^{k+1}$, we denote by $\repbis{\xv}$, the representative of $\class{\xv}$ whose first non-zero symbol is $1$. Thus, for $\sL_1$, we produce successively the representatives $\repbis{\xv_1}$ that have a weight $p$ on the first half and zero elsewhere using exactly the same trick as Peters (except we fix the first non-zero symbol to $1$ and that there is no syndrome to add). So we can compute successively $\Pm' \repbis{\xv_1}$ by only adding one column or two consecutive columns (except for the first element that needs $2p-2$ column additions). The single column additions allow to browse all the vectors of a same support and the two columns additions allow to move from one support to another. The two columns additions can actually be replaced by single column additions if we pre-compute all the sums of two consecutive columns in $\Pm'$. This costs $k+1$ column additions. By fixing the first non zero symbol to one, we browse only one representative per equivalence class. However, it is not the good representative that we defined in Definition \ref{def:repr}. But it is quite easy to find the factor $\alpha \in \Fq^*$ such that $\rep{\xv_1} = \alpha \repbis{\xv_1}$ by multiplying one column by a scalar. Then, we do not compute $\rep{\xv_1}$ but we save the pair $(\alpha,\repbis{\xv_1})$ instead.
We proceed similarly for the list $\sL_2$. So finally, the cost of producing $\sL_1$ and $\sL_2$ is
\begin{eqnarray}
T_{\mathrm{lists}} & = & \ell (k+1) + \ell (2p - 2) + 2 \ell  \left( \binom{\lfloor\frac{k+1}{2}\rfloor}{p} + \binom{k + 1 - \lfloor\frac{k+1}{2}\rfloor}{p} \right)(q-1)^{p-1} \\
 & = & \ell \left( k + 2p - 1 + 2\left( \binom{\lfloor\frac{k+1}{2}\rfloor}{p} + \binom{k + 1 - \lfloor\frac{k+1}{2}\rfloor}{p} \right)(q-1)^{p-1} \right) \label{eq:TListsProj}
\end{eqnarray}

To deal with the collisions we first need to count them. 
On average, there are
\begin{eqnarray}
N_{\mathrm{collision}} & = & \dfrac{\binom{\lfloor\frac{k+1}{2}\rfloor}{p} \binom{k + 1 - \lfloor\frac{k+1}{2}\rfloor}{p} (q-1)^{2p-2}}{1 + (q^{\ell} -1)/(q-1)} \\
 & & \;\;\;\;\;\;\;\;\;\; \cdot \left((q-1)\prob{}{\Pm' \rep{\xv_1} = \Ov} + \prob{}{\Pm' \rep{\xv_1} \neq \Ov} \right) \\
& = & \dfrac{\binom{\lfloor\frac{k+1}{2}\rfloor}{p} \binom{k + 1 - \lfloor\frac{k+1}{2}\rfloor}{p} (q-1)^{2p-2}}{1 + (q^{\ell} -1)/(q-1)} \cdot \left(\frac{q-1}{q^{\ell}} + 1 - \frac{1}{q^{\ell}} \right) \\
& = & \dfrac{\binom{\lfloor\frac{k+1}{2}\rfloor}{p} \binom{k + 1 - \lfloor\frac{k+1}{2}\rfloor}{p} (q-1)^{2p-1}}{q^{\ell}}
\end{eqnarray}
For each collision, we must first apply the scalar multiplication to change the representative of the equivalence class. This costs $2p$ multiplications. Then we can apply the same trick as Peters and check for only $\frac{q}{q-1}(t-2p+1)$ rows on average. The cost to treat a row is $2p$ additions and $2p\frac{q-2}{q-1}$ multiplications. So the cost to check all the candidates (each coming from a collision) is
\begin{equation}
T_{\mathrm{check}} = \left(2p + \tfrac{q}{q-1}(t-2p+1) 2p\left( 1 + \tfrac{q-2}{q-1}\right)\right) \cdot N_{\mathrm{collision}}.\label{eq:TCheckProj}
\end{equation}
\end{proof}

 \section{The $d$-split decoding problem}
\label{sec:dsplit}

The security of \SDitH is actually based on the $d$-split decoding problem. Before stating this problem, let us bring in the following notation:
\begin{notation}
Let $\vv \in \Fq^n$. For $d$ that divides $n$ and for all $i \in \IInt{1}{d}$, we denote by $\vv_{[i]}$ the $i^{\text{th}}$ piece of $\vv$ of length $\frac{n}{d}$. More formally,
\begin{equation}
\vv_{[i]} \eqdef \vv_{\IInt{(i-1) \frac{n}{d} + 1}{i\frac{n}{d}}} \eqdef (v_j)_{j \in \IInt{(i-1) \frac{n}{d} + 1}{i\frac{n}{d}}}
\end{equation}
\end{notation}

Then the $d$-split syndrome decoding problem can be stated as follows:

\begin{pb}[$d$-split Syndrome Decoding $\SD(d,\Hm,\sv,t)$]
\label{pb:SD-d}
Given a matrix $\Hm \in \Fq^{(n-k) \times n}$, a syndrome $\sv \in \Fq^{n-k}$ and a distance $t \in \IInt{0}{n}$, find a vector $\ev \in \Fq^n$ such that $\Hm \ev = \sv$ and $\forall i \in \IInt{1}{d}$, $\hw{\ev_{[i]}} = t/d$.
\end{pb}

The $d$-split decoding problem is quite similar to the standard decoding problem but with the additional constraint that the error weight must be regularly distributed over $d$ blocks. Note that the syndrome $\sv$ was actually produced using an injected solution $\widetilde{\ev}$ whose weight is precisely regularly distributed. In other word, there exists at least one $\widetilde{\ev} \in \Fq^n$ such that $\sv = \Hm \widetilde{\ev}$ and $\forall i \in \IInt{1}{d}$, $\hw{\widetilde{\ev}_{[i]}} = t/d$.

In the \SDitH specifications, the hardness of the $d$-split syndrome decoding problem is lower bounded by a quantity that depends on the complexity to solve the standard decoding problem. This lower bound is based on the following theorem:

\begin{theorem}[\cite{FJR22a}]
\label{th:proba_part_dsplit}
Let $\Hm$ be drawn uniformly at random in $\Fq^{(n-k) \times n}$ and let $\sv \in \Fq^n$.
If an algorithm can find any particular solution of the $d$-split syndrome decoding problem $\SD(d, \Hm, \sv, t)$ in time $T$ with probability $\varepsilon_d$, then there is an algorithm that can find any particular solution of the syndrome decoding problem $\SD(1, \Hm, \sv, t)$ in time $T$ with probability $\varepsilon_1$ with
\begin{equation} \label{eq:no_tight_lb}
\varepsilon_1 \geq \frac{\binom{n/d}{t/d}^d}{\binom{n}{t}} \varepsilon_d
\end{equation} 
\end{theorem}

Using Theorem \ref{th:proba_part_dsplit}, it can be argued that the best average complexity to solve $\SD(d, \Hm, \sv, t)$ cannot be lower than ${\binom{n/d}{t/d}^d}/{\binom{n}{t}}$ times the best average complexity to solve $\SD(1, \Hm, \sv, t)$. \SDitH measures the difficulty of the $d$-split decoding problem with this lower bound on the complexity 
together with the complexity of the best known attack on $\SD(1, \Hm, \sv, t)$. However, this bound is not tight and  gives optimistic security levels. Indeed, it is considered here that we are only looking for a particular solution; the number of solutions to the problem is not taken into account. But recall that if there are many solutions, we just want to find one of them. We therefore state the following theorem that gives a more precise lower bound on the complexity we can expect to achieve:

\begin{theorem}
\label{th:proba_dsplit}
Let $\Hm$ be drawn uniformly at random in $\Fq^{(n-k) \times n}$ and let $\sv \in \Fq^n$.
If an algorithm can solve $\SD(d, \Hm, \sv, t)$ in time $T$ with probability $p_d$, then there is an algorithm that can solve $\SD(1, \Hm, \sv, t)$ in time $T$ with probability $p_1$ with
\begin{equation} \label{eq:tight_lb}
p_d \leq 1 - \left(1 - \frac{\binom{n}{t}}{\binom{n/d}{t/d}^d} \left( 1 - (1-p_1)^{1/N_1} \right)\right)^{N_d} \approx  \frac{\binom{n}{t}}{\binom{n/d}{t/d}^d} \frac{N_d}{N_1} p_1
\end{equation}
where $N_1$ is the expected number of solutions to the problem $\SD(1, \Hm, \sv, t)$ and $N_d$ is the expected number of solutions to the problem $\SD(d, \Hm, \sv, t)$.
\end{theorem}

\begin{proof}
Let $\sA_d$ be an algorithm which finds a particular solution in $\SD(d, \Hm, \sv, t)$ in times $T$ with probability $\varepsilon_d$. From Theorem \ref{th:proba_part_dsplit}, there exists an algorithm $\sA_1$ which finds a particular solution in $\SD(1, \Hm, \sv, t)$ in times $T$ with probability $\varepsilon_1$ with 
$$
\begin{array}{lcl}
1 - \left(1-\varepsilon_d\right)^{N_d} & \leq & 1 - \left(1 - \frac{\binom{n}{t}}{\binom{n/d}{t/d}^d}\varepsilon_1\right)^{N_d}
\end{array}
$$
We end the proof by noticing that 
$$
p_d  = 1 - \left(1-\varepsilon_d\right)^{N_d} \;\;\;\;\; \text{ and }\;\;\;\;\; p_1  = 1 - \left(1-\varepsilon_1\right)^{N_1}
$$
\end{proof}

In Theorem \ref{th:proba_dsplit}, since we address the $d$-split syndrome decoding problem where $\sv$ is the syndrome of a $d$-split error vector of weight $t$, we have:
\begin{equation}
N_d = 1 + \frac{\binom{n/d}{t/d}^d(q-1)^t - 1}{q^{n-k}}.
\end{equation}
Moreover, the algorithm $\sA_1$ consists essentially in repeating the algorithm $\sA_d$ by permuting the code randomly so this algorithm solves the standard decoding problem where $\sv$ is the syndrome of any $t$-weight error vector. So we have:
\begin{equation}
N_1 = 1 + \frac{\binom{n}{t}^d(q-1)^t - 1}{q^{n-k}}.
\end{equation}

\begin{remark}
Note that when $t$ is smaller than the Gilbert-Varshamov distance, then the only solution to $\SD(1,\Hm,\sv,t)$ is the injected solution and has a probability ${\binom{n/d}{t/d}^d}/{\binom{n}{t}}$ to be a solution for $\SD(d,\Hm,\sv,t)$. So we get the same lower bound as in \SDitH since $p_1 \approx \varepsilon_1$ and $p_d \approx \varepsilon_d$. On the contrary, if $t$ is such that we have many solutions, then the injected solution has little impact and so we only have $p_d \leq p_1 (1+\oo{1})$. Note that when the number of solutions is large, $p_d$ is close to $p_1$ because it is simpler to find a particular solution to the $d$-split decoding problem but there are also  less solutions in proportion.\\
\end{remark}

\textbf{Adapting the projective Stern algorithm for d-split.} Theorem \ref{th:proba_dsplit} induces a lower bound on the complexity of $d$-split decoding, but we cannot guarantee that it is actually possible to reach this bound. It is possible to give an actual algorithm to solve the $d$-split decoding problem. Indeed, we can apply the reduction of Section \ref{sec:reduction} and adapt our projective Stern algorithm to take into account the regularity of the weight of the solution we are looking for. More precisely, at each iteration of Algorithm \ref{algo:SternProj}, one can choose the information set as 
\begin{equation}
I \eqdef I_1 \cup \dots \cup I_d
\end{equation}
where each $I_i$ is a subset of $\IInt{(i-1)\frac{n}{d} + 1}{i\frac{n}{d}}$ of size $\frac{k}{d}$. At each iteration, we bet that at least one sought solution $\ev$ is such that for all $i \in \IInt{1}{d}$, $\hw{\ev_{I_i}}=\frac{p}{d}$.

Note that, in the context of \SDitH, when adapting projective Stern to $d$-split, the optimal parameter $p$ increases: it goes from $1$ to $2$. So the cost to produce the lists $\sL_1$ and $\sL_2$ increases quadratically. Consequently, the Gaussian elimination step becomes negligible, especially if we factorize it. It follows that the Canteaut-Chabaud technique is not relevant because it increases the number of iterations but it does not substantially reduce their cost. This is why finally, for $d > 1$, we do not use the Canteaut-Chabaud's trick.
Appendix \ref{appendix:dsplit}, gives the formulas for the $2$-split projective Stern algorithm. \section{Application to SDitH} \label{sec:sdith}
In this section, we analyze the security of \SDitH and we compare our results with those given in the $\SDitH$ specifications document \cite{AFGGHJJPRRY23}. In the context of the NIST competition\footnote{\url{https://csrc.nist.gov/projects/post-quantum-cryptography}}, the authors tried to reach different security levels: \\
\indent - \textbf{category I:} at least $143$ bits of security ($\approx$ \textsf{AES}-$128$);\\
\indent - \textbf{category III:} at least $207$ bits of security ($\approx$ \textsf{AES}-$192$);\\
\indent - \textbf{category V:} at least $272$ bits of security ($\approx$ \textsf{AES}-$256$).\\
Table \ref{tab:params_SDitHv10} summarizes the parameters proposed in the NIST submission \cite{AFGGHJJPRRY23}  of \SDitH to achieve the above security levels. 
\begin{table}[h]
\begin{center}
\caption{\label{tab:params_SDitHv10}
Parameters of  \SDitH v1.0 for various security levels.}
\newcolumntype{C}[1]{>{\centering\arraybackslash}p{#1}}
\scalebox{0.9}{
\begin{tabular}{|C{3.6cm}||C{1.4cm}|C{1.4cm}||C{0.7cm}|C{0.7cm}|C{0.7cm}|C{0.7cm}|C{0.5cm}|}
\hline
\bf \begin{tabular}{c} Parameter \\ sets \end{tabular} & \multicolumn{2}{c||}{\bf \begin{tabular}{c} NIST \\ recommendations \end{tabular}} & \multicolumn{5}{c|}{\bf \begin{tabular}{c} $d$-split SD\\ parameters \end{tabular}} \\
 & category & \begin{tabular}{c} target \\ security \end{tabular} & $q$ & $n$ & $k$ & $t$ & $d$  \\
\hline
\hline
\texttt{SDitH\_L1\_gf256\_v1.0} & I & $143$ bits & $256$ & $230$ & $126$ & $79$ & $1$ \\
\hline
\texttt{SDitH\_L1\_gf251\_v1.0} & I & $143$ bits  & $251$ & $230$ & $126$ & $79$ & $1$ \\
\hline
\texttt{SDitH\_L3\_gf256\_v1.0} & III & $207$ bits  & $256$ & $352$ & $193$ & $120$ & $2$ \\
\hline
\texttt{SDitH\_L3\_gf251\_v1.0} & III & $207$ bits  & $251$ & $352$ & $193$ & $120$ & $2$ \\
\hline
\texttt{SDitH\_L5\_gf256\_v1.0} & V & $272$ bits  & $256$ & $480$ & $278$ & $150$ & $2$ \\
\hline
\texttt{SDitH\_L5\_gf251\_v1.0} & V & $272$ bits  & $251$ & $480$ & $278$ & $150$ & $2$ \\
\hline
\end{tabular}
}
\end{center}
\end{table}

In the \SDitH specifications, it is considered that the best algorithm to solve the decoding problem is Peters' version of Stern's algorithm of Section \ref{sec:stern_peters}. However, there is a mistake in \SDitH v1.0: it is considered that there is only one solution to the syndrome decoding problem when in fact there are several hundred
for the parameters that have been chosen.
Moreover, Theorem \ref{th:proba_part_dsplit} which lower bounds the complexity of  solving the $d$-split is not tight when there are many solutions. In Table \ref{tab:Results_SDiTHv0.1} we give (i) the results claimed in the specifications of \SDitH v1.0, (ii) the lower bound on the security when considering the multiple solutions and a tighter lower bound obtained from Theorem \ref{th:proba_dsplit}, (iii) the lower bound on the 
security we achieve with our projective Stern decoding used in conjunction with Theorem \ref{th:proba_dsplit} and (iv) the security we achieve with the actual $d$-split Stern's algorithm described at the end of Section \ref{sec:dsplit}. The security is expressed in number of bits.
Note that the last column corresponds to an actual attack on the scheme. Comparing (ii) and (iii), we can see in particular that the projective method is responsible for the loss of around $5$ bits of security. In summary:
\begin{itemize}
\item
even by improving the lower bound of \cite{AFGGHJJPRRY23}, this methodology for proving the security fails to meet the NIST requirements by around $11$-$14$ bits, (column (iii))
\item
 there is an actual attack on the scheme showing that its complexity is below the NIST requirements by around $9$-$14$ bits, (column (iv)).
 \end{itemize} 
\begin{table}[h]
\begin{center}
\vspace{-1cm}
\caption{\label{tab:Results_SDiTHv0.1}
Security level of \SDitH v1.0.}
\newcolumntype{C}[1]{>{\centering\arraybackslash}p{#1}}
\scalebox{0.86}{
\begin{tabular}{|C{3.3cm}||C{0.3cm}|C{0.3cm}|C{1.3cm}||C{0.3cm}|C{0.3cm}|C{1.7cm}||C{0.3cm}|C{0.3cm}|C{0.3cm}|C{1.6cm}||C{0.3cm}|C{0.3cm}|C{1.6cm}|}
\hline
\bf \begin{tabular}{c} Parameter \\ sets \end{tabular} &  \multicolumn{3}{c||}{\bf  \begin{tabular}{c} (i) Claimed \\in the \\specification \\ document \end{tabular}} & \multicolumn{3}{c||}{\bf \begin{tabular}{c} (ii) Correction \\from Section \ref{sec:stern_peters} \\ and \\ Theorem \ref{th:proba_dsplit} \end{tabular}} & \multicolumn{4}{c||}{\bf  \begin{tabular}{c} (iii) Projective \\Stern and \\Theorem \ref{th:proba_dsplit} \end{tabular}} & \multicolumn{3}{c|}{\bf  \begin{tabular}{c} (iv) $d$-split \\ projective \\Stern \end{tabular}} \\
 & $p$ & $\ell$ & security & $p$ & $\ell$ & security  & $p$ &$\ell$ & $c$ & security & $p$ &$\ell$ & security \\
\hline
\hline
\texttt{SDitH\_L1\_gf256\_v1.0} & $1$ & $2$ & $\geq 143.46$ & $1$ & $2$& $\geq 135.29$ & $1$ & $2$ & $1$ & $\bf \geq 129.23$ & $1$ & $2$ & $\bf 129.23$ \\
\hline
\texttt{SDitH\_L1\_gf251\_v1.0}  & $1$ & $2$ & $\geq 143.45$ & $1$ & $2$ &$\geq 134.58$ & $1$ & $2$ & $1$ &  $\bf \geq 128.52$ & $1$ & $2$ &  $\bf 128.52$\\
\hline
\texttt{SDitH\_L3\_gf256\_v1.0} & $2$ & $5$ & $\geq 207.67$ & $2$ & $5$ & $\geq 202.43$ &$1$ &$2$ & $1$ &  $\bf \geq 196.76$  &$2$ & $4$ &  $\bf 199.30$\\
\hline
\texttt{SDitH\_L3\_gf251\_v1.0}  & $2$ & $5$ & $\geq 207.61$ & $2$ & $5$ & $\geq 201.30$ & $1$ & $2$ & $1$ &  $\bf \geq 195.68$ & $2$ & $4$ &  $\bf 198.19$ \\
\hline
\texttt{SDitH\_L5\_gf256\_v1.0} & $2$ & $5$ & $\geq 272.35$ & $2$ & $5$ & $\geq 267.40$ & $1$ & $2$ & $1$ &  $\bf \geq 262.63$ & $2$ & $4$ &  $\bf 264.30$ \\
\hline
\texttt{SDitH\_L5\_gf251\_v1.0} & $2$ & $5$ & $\geq 272.29$ & $2$ & $5$ & $\geq 265.91$ & $1$ & $2$ & $1$ &  $\bf \geq 261.19$ & $2$ & $4$ &  $\bf 262.84$\\
\hline
\end{tabular}
}
\end{center}
\end{table}
Recently, after we communicated preliminary results in \cite{CHT23}, the \SDitH designers proposed a new version 1.1 in 
\url{https://groups.google.com/a/list.nist.gov/g/pqc-forum/c/OOnB655mCN8/m/rL4bPD20AAAJ} 
with updated parameters. Table \ref{tab:params_SDitHv11} presents the new parameters.
\begin{table}[h]
\begin{center}
\caption{\label{tab:params_SDitHv11}
Parameters of the $d$-split decoding problem in \SDitH v1.1 for various security levels.}
\newcolumntype{C}[1]{>{\centering\arraybackslash}p{#1}}
\scalebox{0.9}{
\begin{tabular}{|C{3.6cm}||C{1.4cm}|C{1.4cm}||C{0.7cm}|C{0.7cm}|C{0.7cm}|C{0.7cm}|C{0.5cm}|}
\hline
\bf \begin{tabular}{c} Parameter \\ sets \end{tabular} & \multicolumn{2}{c||}{\bf \begin{tabular}{c} NIST \\ recommendations \end{tabular}} & \multicolumn{5}{c|}{\bf \begin{tabular}{c} $d$-split SD\\ parameters \end{tabular}} \\
 & category & \begin{tabular}{c} target \\ security \end{tabular} & $q$ & $n$ & $k$ & $t$ & $d$  \\
\hline
\hline
\texttt{SDitH\_L1\_gf256\_v1.1} & I & $143$ bits & $256$ & $242$ & $126$ & $87$ & $1$ \\
\hline
\texttt{SDitH\_L1\_gf251\_v1.1} & I & $143$ bits  & $251$ & $242$ & $126$ & $87$  & $1$ \\
\hline
\texttt{SDitH\_L3\_gf256\_v1.1} & III & $207$ bits  & $256$ & $376$ & $220$ &$114$ & $2$ \\
\hline
\texttt{SDitH\_L3\_gf251\_v1.1} & III & $207$ bits  & $251$ & $376$ & $220$ &$114$ & $2$ \\
\hline
\texttt{SDitH\_L5\_gf256\_v1.1} & V & $272$ bits  & $256$ & $494$ &$ 282$ & $156$ & $2$ \\
\hline
\texttt{SDitH\_L5\_gf251\_v1.1} & V & $272$ bits  & $251$ & $494$ &$ 282$ & $156$ & $2$ \\
\hline
\end{tabular}
}
\end{center}
\end{table}
Then, Table \ref{tab:Results_SDiTHv1.1} compares the security claimed in \SDitH v1.1 and our projective Stern decoder on the updated parameters. 
Our projective Stern's algorithm is around $2^5$ times faster than the complexity claimed in \SDitH v1.1. However, the \SDitH designers took a margin of error of $4$ bits compared to the NIST recommendations and therefore our attack is only one bit under the NIST recommendations for the category I  parameters. 
For the other parameters, the lower bound methodology of \cite{AFGGHJJPRRY23} (even after improvement) fails to meet the NIST criterion by about one bit in all cases (column (iii)). 
It remains to see whether the attack on the $d$-split version can be improved (this is used in category III and V parameters) because our corresponding attack (column (iv)) is just
one bit above the required security level. The \textsf{SageMath} program which made it possible to compute the results is available on \url{https://github.com/kevin-carrier/SDitH_security}.
\begin{table}[h]
\begin{center}
\caption{\label{tab:Results_SDiTHv1.1}
Security level of \SDitH v1.1.}
\newcolumntype{C}[1]{>{\centering\arraybackslash}p{#1}}
\scalebox{0.9}{
\begin{tabular}{|C{3.4cm}||C{1.4cm}||C{0.3cm}|C{0.3cm}|C{1.65cm}||C{0.3cm}|C{0.3cm}|C{0.3cm}|C{1.8cm}||C{0.3cm}|C{0.3cm}|C{1.65cm}|}
\hline
\bf \begin{tabular}{c} Parameter \\ sets \end{tabular} & \bf \begin{tabular}{c} Target \\ security  \end{tabular} & \multicolumn{3}{c||}{\bf \begin{tabular}{c} (i-ii) Claimed \\in the \\specification \\ document \end{tabular}} & \multicolumn{4}{c||}{\bf \begin{tabular}{c} (iii) Projective \\Stern and \\Theorem \ref{th:proba_dsplit} \end{tabular}} & \multicolumn{3}{c|}{\bf \begin{tabular}{c} (iv) $d$-split \\ projective \\Stern \end{tabular}} \\
 &  & $p$ & $\ell$ & security  & $p$ &$\ell$& $c$&  security  & $p$ &$\ell$&  security \\
\hline
\hline
\texttt{SDitH\_L1\_gf256\_v1.1} & $143$ & $1$ & $2$& $\geq 147.73$ & $1$ & $2$ & $1$ & $\bf \geq 141.54$ & $1$ & $2$ & $\bf 141.54$ \\
\hline
\texttt{SDitH\_L1\_gf251\_v1.1} &  $143$ & $1$ & $2$ &$\geq 147.72$ & $1$ & $2$ & $1$ & $\bf \geq 141.54$ & $1$ & $2$ & $\bf 141.54$\\
\hline
\texttt{SDitH\_L3\_gf256\_v1.1}  & $207$ & $2$ & $5$ & $\geq 211.05$ &$1$ &$2$ & $1$ & $\bf \geq 205.59$ & $2$ & $4$ & $\bf 207.90$\\
\hline
\texttt{SDitH\_L3\_gf251\_v1.1}  & $207$ & $2$ & $5$ & $\geq 210.99$ & $1$ & $2$ & $1$ & $\bf \geq 205.59$ & $2$ & $4$ & $\bf 207.86$ \\
\hline
\texttt{SDitH\_L5\_gf256\_v1.1} &  $272$ & $2$ & $5$ & $\geq 276.33$ & $1$ & $2$ & $1$ & $\bf \geq 271.69$ & $2$ & $4$ & $\bf 273.26$ \\
\hline
\texttt{SDitH\_L5\_gf251\_v1.1} & $272$ & $2$ & $5$ & $\geq 276.28$ & $1$ & $2$ & $1$ & $\bf \geq 271.68$ & $2$ & $4$ & $\bf 273.23$\\
\hline
\end{tabular}
}
\end{center}
\end{table}

\newcommand{\etalchar}[1]{$^{#1}$}

\appendix
\section{Formulas for the $2$-split projective Stern algorithm} 
\label{appendix:dsplit}

Here we address the $2$-split decoding problem which is one of the instances of $\SDitH$. In Section \ref{sec:dsplit}, we explained how to adapt our projective Stern's algorithm to $d$-split by splitting the support of the error into $d$ equal parts and looking for an error $\ev$ of weight $t/d$ on each of the part.

We just give the formulas which differ from those we presented for the standard version $1$-split in Section \ref{sec:projection}. In particular, for the $2$-split version of our projective Stern's algorithm, Equations \eqref{eq:qin}, \eqref{eq:Nout_tmp}, \eqref{eq:TGaussProj} \eqref{eq:Nsol_proj}, \eqref{eq:list1_proj} and \eqref{eq:list2_proj} become respectively:

\begin{equation}
\begin{array}{lcl}
q_{\mathrm{in}} & = & \dfrac{ \binom{\lfloor (k+1)/4 \rfloor}{\lfloor p/2 \rfloor} \binom{\lfloor (k+1)/2 \rfloor - \lfloor (k+1)/4 \rfloor}{p - \lfloor p/2 \rfloor} }{ \binom{\lfloor (k + 1)/2 \rfloor}{p} } \\[0.5cm]
& &  \;\;\;\;\;\;\;\;\;\; \cdot \; \dfrac{ \binom{\lfloor (k+1)/2 \rfloor - \lfloor (k+1)/4 \rfloor}{\lfloor p/2 \rfloor} \binom{k + 1 - 2\lfloor (k+1)/2 \rfloor + \lfloor (k+1)/4 \rfloor}{p - \lfloor p/2 \rfloor} }{ \binom{k + 1 - \lfloor (k + 1)/2 \rfloor}{p} } \\[0.5cm]
& & \;\;\;\;\;\;\;\;\;\;\;\;\;\;\; \cdot \; \dfrac{ \binom{\lfloor (n - k - 1)/2 \rfloor - \lfloor \ell/2 \rfloor}{t/2 - p} \binom{n - k - 1 - \lfloor (n - k - 1)/2 \rfloor  - \ell + \lfloor \ell/2 \rfloor}{t/2 - p} }{ \binom{\lfloor (n - k - 1)/2 \rfloor}{t/2 - p} \binom{n - k - 1 - \lfloor (n - k - 1)/2 \rfloor}{t/2 - p}}
\end{array}
\end{equation}  

\begin{equation}
N_{\mathrm{out}}^{\mathrm{tmp}} = \dfrac{ \binom{\lfloor n/2 \rfloor}{ \lfloor t/2 \rfloor}^2 }{ \binom{\lfloor (k+1)/2 \rfloor}{p} \binom{k+1-\lfloor (k+1)/2 \rfloor}{p} \binom{\lfloor (n-k-1)/2 \rfloor}{t/2-p} \binom{n-k-1-\lfloor (n-k-1)/2 \rfloor}{t/2-p} }
\end{equation}

\begin{equation}
T_{\mathrm{Gauss}} = 2 (n-k-1) \sum_{i = 1}^{n-k-1} (n - i + 1) = (n-k-1)^2 (n + k + 2)
\end{equation}

\begin{equation}
N_{\mathrm{sol}} = 1 + \frac{\binom{n/2}{t/2}^2 (q-1)^{t-1} - 1}{q^{n-k-1}}
\end{equation}

\begin{eqnarray}
L_1  & = & \textstyle{\binom{\lfloor (k+1)/4 \rfloor}{\lfloor p/2 \rfloor} \binom{\lfloor (k+1)/2 \rfloor - \lfloor (k+1)/4 \rfloor}{p - \lfloor p/2 \rfloor} (q-1)^{p-1}}\\[0.5cm]
L_2 & = & \textstyle{\binom{\lfloor (k+1)/2 \rfloor - \lfloor (k+1)/4 \rfloor}{\lfloor p/2 \rfloor} \binom{k + 1 - 2\lfloor (k+1)/2 \rfloor + \lfloor (k+1)/4 \rfloor}{p - \lfloor p/2 \rfloor} (q-1)^{p-1}}
\end{eqnarray} \end{document}